\documentclass[pra,aps,twocolumn,reprint,showpacs,amsmath,amssymb,groupedaddress,floatfix, superscriptaddress,longbibliography]{revtex4-1}
\usepackage{general_latex}
\usepackage[justification=centerfirst]{caption}
\usepackage{dsfont}

\usepackage{amsmath}
\usepackage{amssymb}
\usepackage{amsthm}
\theoremstyle{plain}

\newtheoremstyle{mystyle}
  {}                                      
  {}                                      
  {\itshape}                              
  {}                                      
  {\bfseries}                             
  {.}                                     
  { }                                     
  {\thmname{#1}\thmnumber{ #2}\thmnote{ (#3)}}%

\theoremstyle{mystyle}

\newtheorem*{theorem*}{Theorem}

\begin{document}
\title{Implementation-independent sufficient condition of the Knill-Laflamme type 
for the autonomous protection of logical qudits   
by strong engineered dissipation}
\author{Jae-Mo Lihm}
\affiliation{Department of Physics and Astronomy, Seoul National University,  \\
Center for Theoretical Physics, 08826 Seoul, Korea}
\author{Kyungjoo Noh}
\affiliation{Yale Quantum Institute, Yale University, New Haven, Connecticut 06520, USA}
\author{Uwe R. Fischer}
\affiliation{Center for Theoretical Physics, Department of Physics and Astronomy, Seoul National University, 08826 Seoul, Korea}
\begin{abstract}
Autonomous quantum error correction utilizes the engineered coupling of a quantum system to a dissipative ancilla to protect   
quantum logical states from decoherence. We show that the Knill-Laflamme condition, stating that the environmental error operators should act trivially on a subspace, which then becomes the code subspace, is sufficient for logical qudits to be protected against Markovian noise. 
It is proven that the error caused by the total Lindbladian evolution in the code subspace can be suppressed up to very long times in the limit of large engineered dissipation, by explicitly 
deriving how the error scales with both time and engineered dissipation strength.
To demonstrate the potential of our approach for applications, we implement our general theory with  binomial codes, a class of bosonic error-correcting codes, and outline how they can be implemented in a fully autonomous manner to protect against photon loss in a microwave cavity. 
\end{abstract}
\maketitle
\section{Introduction}  
In the majority of practically realized experiments, 
the 
unavoidable coupling to an environment leads to non-unitary evolution
 and introduces noise and dissipation into the system \cite{breuer2002theory}. This is often regarded as detrimental and hindering the control of the experiment. 
In the last two decades however, the at first sight counterintuitive fact has been established that 
dissipation can in fact {\em support} various quantum information processing tasks, such as 
quantum computation \cite{Braun2000,Verstraete,Zanardi2014,Marshall2016}, 
entanglement generation \cite{Kraus2008,Krauter2011,Lin2013,Shankar2013}, and the preparation of quantum many-body states \cite{Diehl,Diehl2011}. 

Quantum error correction (QEC) is a method of correcting for effects of 
noise in quantum computers \cite{Unruh} by quantum control \cite{Shor1995,Cory1998,Chiaverini,Schindler2011,Kelly,Ofek,Linke}. 
Conventional measurement-based QEC utilizes periodic measurement of error syndromes and an appropriate unitary evolution that eliminates the error.
This measurement-based QEC scheme can correct errors 
iff encoded logical quantum states satisfy the Knill-Laflamme condition \cite{Knill1997,Knill2000}, 
which predicates that errors can be corrected iff an 
environment 
causing the error does not gain any information about the encoded states \cite{BenyAQEC}. 

Autonomous QEC (AutoQEC) 
avoids the use of active syndrome measurements and real-time feedback control to achieve QEC, 
and instead
takes advantage of passive {\em strong engineered} dissipation, which restores the system into the original code subspace manifold. Pioneering works \cite{Paz1998,Ahn2002,Sarovar2005} mostly focused on autonomous implementation of the 3-bit repetition code against single bit-flip error, and then were generalized to all stabilizer codes \cite{Gottesman} in \cite{Oreshkov2013,Hsu2016}. Besides the stabilizer-based qubit codes, many specific AutoQEC schemes tailored to certain experimental platforms were proposed \cite{Kerckhoff2010,
Kerckhoff2011,Leghtas2013,AutoCatTheory,Kapit2015,
Kapit2016,Reiter2017,Cohen2017,Albert2018}. Notably, the simplest version of the cat code proposals \cite{Leghtas2013,AutoCatTheory,Cohen2017} was recently realized experimentally in circuit QED systems \cite{Leghtas853,Touzard2017}.

In what follows, we 
identify an {\em implementation independent}, sufficient condition for AutoQEC, of the Knill-Laflamme type.  In particular, we prove its sufficiency by providing an explicit Markovian engineered dissipation and by deriving a rigorous upper bound of logical error probability in terms of the engineered dissipation strength. In addition, we do not assume any structure of the error-correcting codes, and thus our theory also applies to the codes which are beyond the paradigm of stabilizer-based qubit codes. We illustrate this 
by proposing a fully autonomous implementation of the recently proposed 
binomial codes \cite{BinomialCode}, which protects encoded information against photon loss in a microwave cavity.

\section{AutoQEC Theorem} 
Consider an open quantum system evolving according to a 
Markovian master equation: 
\begin{align}
\frac{\partial \rho}{\partial t} = \mc{L}[\rho] &= -i[H_{\mathrm{eng}}+H_{\mathrm{sys}},\rho] 
\nonumber\\
&\quad + M\sum_{i=1}^{L}D(F_{\mathrm{eng},i})\rho + \sum_{k=1}^{N}D(F_{\mathrm{sys},k})\rho, \label{eq:master}
\end{align}
where $D(A)\rho \coloneqq A\rho A^{\dagger}-\frac{1}{2}(A^{\dagger}A\rho +\rho A^{\dagger}A )$. $H_{\mathrm{sys}}$ and $F_{\mathrm{sys}}$, and $H_{\mathrm{eng}}$ and $F_{\mathrm{eng}}$ are Hamiltonian and Lindbladian operators that are intrinsic to the system and imposed by dissipative engineering, respectively; we set $\hbar =1$. We omit the subscript $\mathrm{sys}$ from the intrinsic Lindblad operators in the following. 
We assume that intrinsic  
dissipation is uncontrollable, while engineered dissipation, in 
Markovian form, can be artificially generated
and increased at will.
In the limit $M\gg 1$ we consider, the engineered 
dissipation strength parametrized by $M$ is much larger than its intrinsic counterpart.
We discuss the practical feasibility of strong Markovian dissipation  
for a concrete implementation 
below.

We seek a condition for the code space $\mathcal{C}\subset \mathcal{H}$ 
($\text{dim}\,\mathcal{H} =n < \infty$, with $\mc{H}$ the total Hilbert space), 
under which there exist an engineered Hamiltonian $H_{\mathrm{eng}}$ and Lindblad jump operators $\{F_{\mathrm{eng},i}\}_{i=1,\cdots,L}$ which autonomously protect the code space against the Markovian error generated by the intrinsic jump operators $\lbrace F_{k}\rbrace_{k=1,\cdots,N}$. 
That a code space is {\em protected} means that for
 sufficiently strong engineered dissipation, 
any logical states in the 
$d$-dimensional code subspace $\mc{C}$ 
are left invariant until a given time $T$ and up to an arbitrarily small error $\epsilon$, 
which is defined in Eq.~\eqref{eq:thm} below. 

We find that the Knill-Laflamme condition \cite{Knill1997,Knill2000}, when 
applied to the set of intrinsic error operators $\mc{E} = \{\mathds{1},
F_k\}_{k=1,...,N}$, 
is sufficient for the autonomous protection of the qudit code subspace 
$\mathcal{C}=\textrm{span}\lbrace |W_{\mu}\rangle\rbrace_{\mu=0,\cdots,d-1}$. The condition states that 
\begin{equation} \label{eq:k-l}
\Pi_\mathcal{C} E^\dagger_{l'} E_l \Pi_\mathcal{C} = c_{l'l} \Pi_\mathcal{C}, 
\end{equation}
for all $E_{l},E_{l'}\in\mathcal{E}$ and $l,l'\in \lbrace 0,...,N\rbrace$, 
where $\langle W_{\mu}|W_{\mu'} \rangle =\delta_{\mu\mu'}$, $\Pi_\mathcal{C} \coloneqq\sum_{\mu=0}^{d-1}|W_{\mu}\rangle \langle W_{\mu}| $ and $c_{k'k}$ are constants. Our main technical result, then, is as follows:

\begin{theorem*}[AutoQEC condition] 
If a code space $\mathcal{C}$ satisfies the Knill-Laflamme condition for the error set $\mc{E} = \{\mathds{1}, F_k\}_{k=1,...,N}$ (see Eq.~\eqref{eq:k-l}), there exists a set of engineered dissipative jump operators $F_{\mathrm{eng},i}$ such that, for any initial density matrix $\rho_{\mathcal{C}}$ in the code space, 
\begin{equation}    \label{eq:thm}
\epsilon(T;\rho_{\mathcal{C}})  \coloneqq\opnorm{ e^{T\mc{L}} \rho_{\mathcal{C}} - \rho_{\mathcal{C}}  } \leq 
\mc{O} \left({\gamma T}/{M}\right), 
\end{equation} 
where $\gamma \coloneqq\sum_{k=1}^N \parallel \!\! F_k \!\! \parallel^2 $ denotes the measure of intrinsic dissipation strength. Throughout, we use $ \parallel \!\! X \!\! \parallel \coloneqq \textrm{sup}_{\opnorm{v}=1}\! \parallel \!\! {Xv}\!\! \parallel$ as the norm of an operator $X$, 
 where $v$ are Hilbert space vectors, while superoperators have 
norm $\opnorm{\mc{Y}} \coloneqq \textrm{sup}_{\opnorm{X}=1} \opnorm{\mc{Y}(X)}$.  
\label{theorem:AutoQEC condition}
\end{theorem*}

\begin{proof}[Summary of the proof]
We first note that the intrinsic Hamiltonian $H_{\mathrm{sys}}$ can be eliminated by setting $H_{\mathrm{eng}}=-H_{\mathrm{sys}}$. Consider a code space $\mathcal{C}$ satisfying the AutoQEC condition in Eq. \eqref{eq:k-l}. The density matrix generated by this code space is represented by the cross-striped areas in Fig.~\ref{fig:schematic}. The intrinsic jump operators $F_{k}$ corrupt the states in $\mathcal{C}$ and move them into the corrupted code space $\mathcal{S}_{\mathrm{ccs}}$ (cf.~green dashed line in Fig. \ref{fig:schematic}). We {explicitly} construct a set of engineered jump operators $F_{\mathrm{eng},i}$ (see Eq.\eqref{eq:Feng_m}) which pump the corrupted states back into the code space $\mathcal{C}$ without loss of coherence (blue double line in Fig. \ref{fig:schematic}). We then show that the $d \times d$ blocks of gray-shaded and cross-striped areas in Fig.~\ref{fig:schematic} form a noiseless subsystem \cite{Zanardi1997,Lidar1998,Knill2000,Kempe2001,Lidar2012}, 
if the intrinsic dissipation operator is limited to act on the code subspace, cf.~Appendix \ref{AppA}. Finally, we show that the remaining dissipation increases the error only to $\ord(\gamma T/M)$, as indicated on the right-hand side of Eq.~\eqref{eq:thm}.
For details on the proof, we refer the reader to Appendix \ref{AppB}.
\end{proof}


\begin{figure}[t]
\centering
\includegraphics[width=0.375\textwidth]{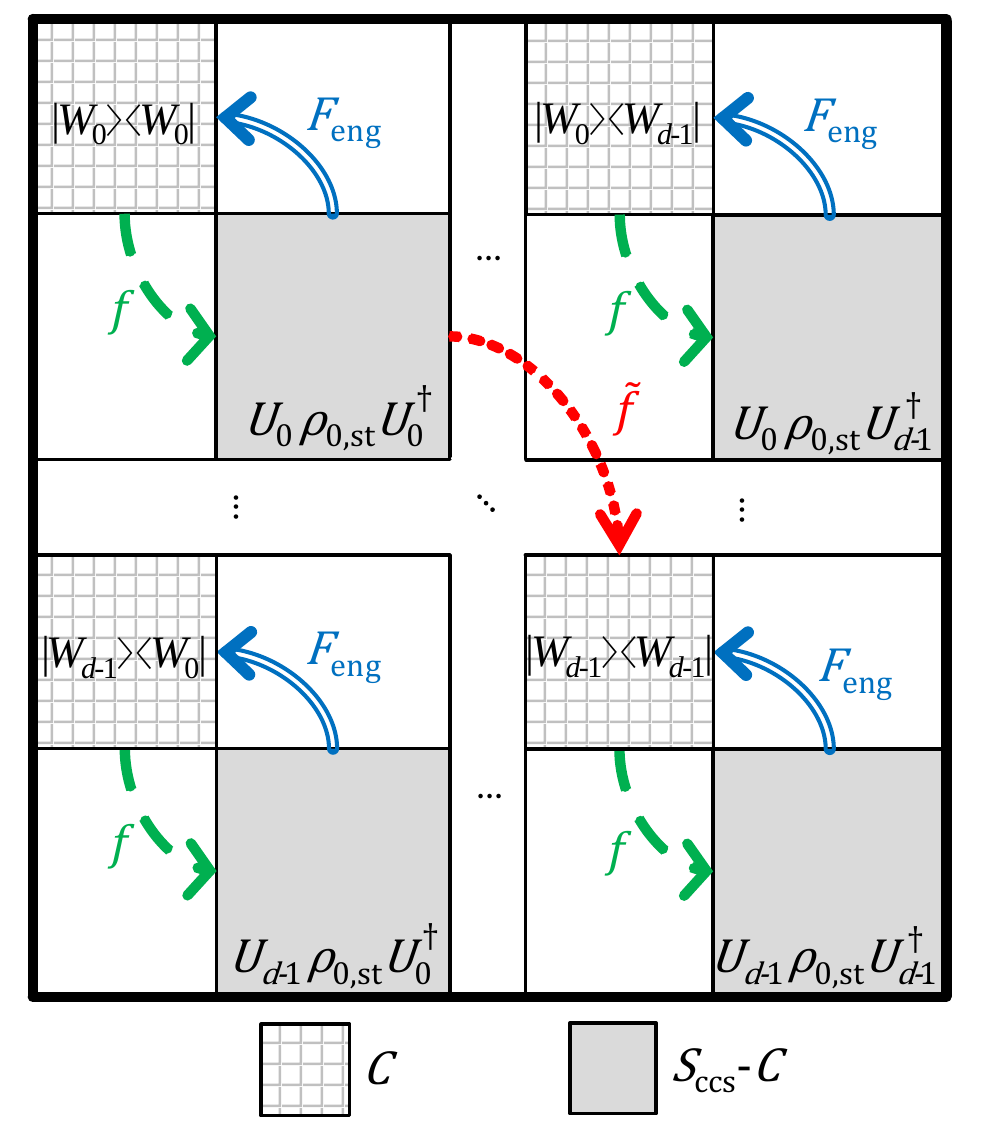}
\caption{{\em Schematic of  
density matrix , Lindblad jump operators, and noiseless subsystem symmetry}. The total Hilbert space is divided into $d\times d$ arrays of $d^{2}$ large blocks.  
The logical code space $\mathcal{C}$ is formed by the
cross-striped areas. The arrows represent dissipative jump operators, and
the gray-shaded areas represent the space of corrupted states, connected via Lindblad operator $f$ (green dashed, Eq.~\eqref{eq:Ff}) to the code states.  
Engineered dissipation $F_{\mathrm{eng}}$ (blue double line, Eq.~\eqref{eq:Feng_m}) pumps the corrupted states back to the code space, and in consequence each large block supports a
steady state $U_{\mu}\rho_{0,\textrm{st}}U^{\dagger}_{\nu}$, and $\mathcal{S}_\textrm{ccs}$ forms a noiseless subsystem. $U_{\mu}$ intertwines between $\mathcal{S}_{0}$ and $\mathcal{S}_{\mu}$, and is defined in Eq.~\eqref{eq:unitary}. 
The remaining intrinsic dissipation $\widetilde{f}$ (red-dotted, Eq.~\eqref{eq:Ff}, for clarity only one such process is shown) breaks this noiseless subsystem symmetry and induces an undesired error.  
}
\label{fig:schematic}
\end{figure}

We observe that the relation \eqref{eq:thm} implies protection of a logical qudit, since one can choose $M \gg \gamma T/\epsilon$ to obtain any desired 
error less than $\epsilon$. Also, we remark that Eq.~\eqref{eq:k-l} was previously identified as the necessary and sufficient condition for the 
elimination of the first order error contribution at short times in \cite{Beny2011}. 
Although related, our approach notably differs from the earlier one 
in that we construct an explicit recovery of the Lindbladian form to achieve AutoQEC (see 
Eq.~\eqref{eq:Feng_m} below), and consider the limit of arbitrarily long times relevant for quantum computing applications.

\section{Design of engineered dissipation} 
Suppose that we are given a code space $\mathcal{C} = \textrm{span} \lbrace |W_{\mu}\rangle \rbrace_{\mu=0,\cdots,d-1}$, satisfying the Knill-Laflamme condition for the error set $\mc{E} = \{\mathds{1}, F_k\}_{k=1,...,N}$ (cf.~Eq.~\eqref{eq:k-l}). Then, we define $d$ subspaces $\mc{S}_\mu$ for each $\mu \in \lbrace 0,\cdots,d-1 \rbrace$:   
\begin{equation} \label{eq:subspace}
    \mc{S}_\mu = \textrm{span}\{ \{|W_{\mu}\rangle\}	\cup \{F_k |W_{\mu}\rangle\ : k = 1,...,N \} \} .
\end{equation}
All $\mc{S}_{\mu}$ are disjoint and have the same dimension, as follows from 
the Knill-Laflamme condition.
We define $|W_{\mu};i\rangle$ such that $\{|W_{\mu}\rangle, |W_{\mu};i\rangle \}$ is an 
orthonormal basis set for $\mc{S}_{\mu}$ ($\mu=0,\cdots,d-1$), where $i=1,\cdots,m-1$ so that the $\mc{S}_{\mu}$ are $m$-dimensional ($m\le N+1$). This basis should be chosen such that $\langle W_{\mu};i|F_{k} |W_{\mu} \rangle$ is $\mu$-independent (cf.~the discussion around Eq.~\eqref{eq:steady states when the error is limited to the code subspace} below)
which, due to the Knill-Laflamme condition Eq.~\eqref{eq:k-l}, can be realized by, e.g.,  
Gram-Schmidt orthonormalization \cite{friedberg2003linear} with $F_k |W_{\mu}\rangle$ of the same iteration order for all $\mu\in\lbrace 0,\cdots,d-1\rbrace$.  
We denote $\mc{S}_{\rm ccs} = \cup_{\mu=0}^{d-1} \mc{S}_{\mu}$ a \textit{corrupted code space}. 
Also, let $\{|\phi_p \rangle\}_{p=0,\cdots,n-md-1}$ represent an orthonormal basis of $\mc{H}-\mc{S}_{\rm ccs} $. 

We design the engineered dissipation as follows: 
\begin{eqnarray}\arraycolsep=1pt\def\arraystretch{2.2}
	F_{\textrm{eng},i} &=&\left\{\begin{array}{l@{\qquad}l}  
	\sum_{\mu=0}^{d-1}|W_{\mu}\rangle\langle W_{\mu};i|  &\! (1 \le i \le m-1)
	\vspace*{0.5em}\\
	|\Phi_{i-m}\rangle\langle \phi_{i-m}| &\!  (m \le i \le L).  
	\end{array}
	\right.
	\label{eq:Feng_m}
\end{eqnarray}
where $L = n-m(d-1)-1$ and $|\Phi_{p}\rangle$ ($p=0,\cdots,n-md-1$) can be any state in $\mathcal{S}_{\textrm{ccs}}$.   
We term the first $m-1$ jump operators \textit{corrective} jumps because they pump erroneous states $|W_{\mu};i\rangle$ back into the code state $|W_{\mu}\rangle$ coherently, 
similarly to measurement-based QEC. The remaining jumps 
(i.e., $F_{\textrm{eng},i}$ with $m\leq i\leq L$) we term \textit{preventive} jump operators: They force states in $\mathcal{H}-\mathcal{S}_{\textrm{ccs}}$ to decay into $\mathcal{S}_{\textrm{ccs}}$, preventing further uncorrectable errors.

The necessity of corrective jump operators is clear from the perspective of measurement-based QEC. However, it might not be immediately obvious whether the preventive jump operators are essential for a successful AutoQEC. In this regard, we emphasize that the preventive jump operators 
are indeed crucial to achieve the desired error scaling in Eq.~\eqref{eq:thm} since otherwise any leakage to $\mathcal{H}-\mathcal{S}_{\textrm{ccs}}$ cannot be recovered (see Fig. \ref{fig:Preventive needed} (a) and the discussion in subsection \ref{subsection:Binomial code against photon loss error} for an illustration).

In addition, while it is essential that the corrective jumps 
are given by the first line of Eq.~\eqref{eq:Feng_m}, preventive jump operators can be chosen arbitrarily as long as  $\text{im}(F_{\textrm{eng},i})\subseteq \mc{S}_{\rm ccs}$ $\forall i\in\lbrace m,\cdots,L\rbrace$ and $\text{ker}(\sum_{i = m}^{L}F_{\textrm{eng},i}^{\dagger}F_{\textrm{eng},i})=\mc{S}_{\rm ccs} $ hold, where $\text{im}(A)$ and $\text{ker}(A)$ represent image space and null space (kernel) of $A$, respectively (cf.~subsection \ref{subsection:Comparison with measurement-based QEC}).

\section{Examples for engineering jumps}
\subsection{3-bit repetition code against bit-flip error} 
In order to show that earlier results can be recovered from our general design, we consider the 3-bit repetition code $|W_{0}\rangle = |000\rangle$, $|W_{1}\rangle=|111\rangle$ 
in the presence of bit flip error: 
\begin{equation}
\frac{d\rho(t)}{dt} = M\gamma\sum_{i}D(F_{\mathrm{eng},i})\rho(t) + \gamma\sum_{k=1}^{3}D(X_{k})\rho(t),  
\end{equation} 
where $X_k$ represents the Pauli operator $X$ (the bit flip by $\sigma_x$) acting on the $k^{\textrm{th}}$ qubit (e.g., $X_{1} = XII$, where $I$ is the identity operator). 
In this case, the AutoQEC error set is given by $\mathcal{E} = \lbrace III,XII,IXI,IIX \rbrace$, and the 3-bit repetition code satisfies the Knill-Laflamme condition for this error set. Then, following our jump operator design principle, we find $|W_{\mu};i\rangle = X_{i}|W_{\mu}\rangle$ ($i=1,2,3$), and thus 
\begin{align}
F_{\mathrm{eng},1} &= |000\rangle\langle 100| + |111\rangle\langle 011| , 
\nonumber\\
F_{\mathrm{eng},2} &= |000\rangle\langle 010| + |111\rangle\langle 101| , 
\nonumber\\
F_{\mathrm{eng},3} &= |000\rangle\langle 001| + |111\rangle\langle 110| , 
\end{align}
which are consistent with the earlier designs in \cite{Paz1998,Ahn2002,Sarovar2005,Reiter2017}. Note that all these jump operators are corrective, and we do not need any preventive jump operators because $\lbrace |W_{\mu}\rangle,|W_{\mu};i\rangle\rbrace_{\mu=0,1}^{i=1,2,3}$ spans the entire Hilbert space (i.e., $\mathcal{H} = \mathcal{S}_{\mathrm{ccs}}$).

\begin{figure*}[t!]
\centering
\includegraphics[width=0.99\textwidth]{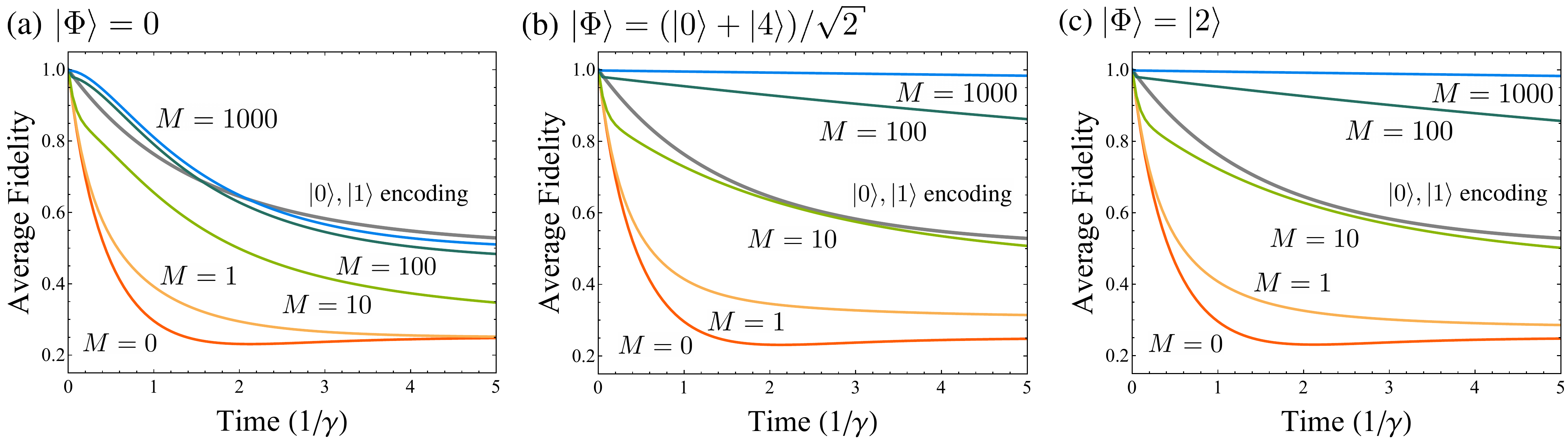}
\caption{{\em Averaged fidelity for various choices of $|\Phi\rangle$ in \eqref{eq:binomial code engineered dissipation}.} 
Fidelity is averaged over $6$ initial input logical states 
$|\psi_{\rm in}\rangle = |W_{0}\rangle,|W_{1}\rangle,(|W_{0}\rangle \pm |W_{1}\rangle)/\sqrt{2}$ and $(|W_{0}\rangle \pm i |W_{1}\rangle)/\sqrt{2}$. The solid gray line represents the fidelity of the lowest energy uncorrected physical qubit (i.e., $|W^{\textrm{phys}}_{0}\rangle = |0\rangle$, $|W^{\textrm{phys}}_{1}\rangle=|1\rangle$) for single photon loss $D(a)$ with decay rate $\gamma$. The other data shows fidelities of the autonomously corrected binomial code with engineered dissipation strength $0\le M\le 1000$ under the same loss rate $\gamma$. The engineered jump operators are given in Eq.~\eqref{eq:binomial code engineered dissipation} and $|\Phi\rangle$ was chosen to be (a) $|\Phi\rangle = 0$, (b) $|\Phi\rangle = (|0\rangle+|4\rangle)/\sqrt{2}$ and (c) $|\Phi\rangle = |2\rangle$. 
}
\label{fig:Preventive needed}
\end{figure*}

\subsection{Binomial code against photon loss error}  
\label{subsection:Binomial code against photon loss error}
In order to exhibit the generality of our approach, we discuss the autonomous protection of the binomial code $|W_{0}\rangle = (|0\rangle+|4\rangle)/\sqrt{2}$, $|W_{1}\rangle=|2\rangle$ (a recently proposed bosonic code; see \cite{BinomialCode}) against the error caused by photon loss, i.e., we consider 
\begin{equation}
\frac{d\rho(t)}{dt} = M\gamma\sum_{i}D(F_{\mathrm{eng},i})\rho(t) + \gamma\sum_{k=1}^{3}D(a)\rho(t),  
\end{equation}   
where $a$ is the annihilation operator of the single bosonic mode. In this case, the AutoQEC error set is given by $\mathcal{E}= \lbrace I,a \rbrace$, and the binomial code satisfies the Knill-Laflamme condition for this error set.

Upon a single photon loss, logical states of the binomial code are corrupted as $a|W_{0}\rangle = \sqrt{2}|3\rangle$ and $a|W_{1}\rangle=\sqrt{2}|1\rangle$, and thus we have $|W_{0};1\rangle = |3\rangle$ and $|W_{1};1\rangle = |1\rangle$. The four states $\lbrace |W_{\mu}\rangle,|W_{\mu};1\rangle \rbrace_{\mu=0,1}$ however do not span the entire Hilbert space $\mathcal{H}_{4}=\textrm{span}\lbrace |0\rangle, \cdots,|4\rangle\rbrace$ due to the residual basis vector $|\phi_{1}\rangle = (|0\rangle-|4\rangle)/\sqrt{2}$. Then, we get the following engineered jump operators from Eq.~\eqref{eq:Feng_m}: 
\begin{align}
F_{\text{eng},1} &= \frac{1}{\sqrt{2}}(|0\rangle+|4\rangle)\langle 3| +|2\rangle\langle 1|, 
\nonumber\\
F_{\text{eng},2} &= |\Phi\rangle\langle 0|-|\Phi\rangle\langle 4|, \label{eq:binomial code engineered dissipation} \end{align}
where $|\Phi\rangle$ can be any state in $\mc{H}_{4}$ orthogonal to $|0\rangle-|4\rangle$. Note that $F_{\text{eng},1}$ is a corrective jump operator while $F_{\text{eng},2}$ is a preventive jump operator.

We stress that the binomial code is not described by a set of stabilizers, and the proposed engineered jump operators in Eq.~\eqref{eq:binomial code engineered dissipation}and, in particular, the preventive jump operator $F_{\mathrm{eng},2}$, are not derivable from the earlier framework laid down, e.g., in \cite{Oreshkov2013,Hsu2016}. To illustrate that the preventive jump operator indeed plays a crucial role, we plot in Fig.~\ref{fig:Preventive needed} the average fidelity of $6$ representative logical states of the autonomously protected binomial code with and without the preventive jump operator. The fidelity is defined as $\langle \psi_{\textrm{in}} | e^{\mathcal{L}t} (|\psi_{\textrm{in}}\rangle\langle\psi_{\textrm{in}}|) |\psi_{\textrm{in}}\rangle$, for an input state $|\psi_{\textrm{in}}\rangle$.

As can be seen from Fig.~\ref{fig:Preventive needed} (a), the autonomously protected binomial code without the preventive jump operator $F_{\mathrm{eng},2}$ (or $|\Phi\rangle=0$) can barely perform comparably to the unprotected physical qubit with the least energy (i.e., $|W_{0}^{\textrm{phys}}\rangle=|0\rangle$, $|W_{1}^{\textrm{phys}}\rangle=|1\rangle$ possessing the longest physical qubit lifetime), even in the $M\rightarrow \infty$ limit. On the other hand, if we add a preventive jump operator (Fig.~\ref{fig:Preventive needed} (b); $|\Phi\rangle = (|0\rangle+|4\rangle)/\sqrt{2}$), the autonomously protected binomial code begins to perform comparably to the physical qubit at $M\simeq 10$ (i.e., at the break-even point; see \cite{Ofek,Linke}), and allows for logical qubit lifetimes much longer than those of the physical qubit if $M \gtrsim 100$, as guaranteed by the AutoQEC Theorem in section \ref{theorem:AutoQEC condition}. Thus, the preventive jump operator $F_{\textrm{eng},2}$ is essential for the superior performance of the binomial code. We also remark that this superior performance (in particular, the scaling $\epsilon(T) = \mathcal{O}(\gamma T /M)$) does not depend on the choice of $|\Phi\rangle$ as long as
$\langle \Phi|\Phi\rangle = \mathcal{O}(1)$ (cf.~Fig.~\ref{fig:Preventive needed} (b,c)).

\subsection{Comparison with measurement-based QEC} 
\label{subsection:Comparison with measurement-based QEC}
The Lindbladian generator $D(A)$ consists of the \textit{jump} term $A\bullet A^{\dagger}$ and the \textit{no-jump evolution} term $-\frac{1}{2} \lbrace A^{\dagger}A,\bullet \rbrace$. We identify that the corrective jump operators suppress the adverse effects of jump-type errors, while the preventive jump operators suppress the effects of the no-jump evolution. Since the jump-type error induces population transfer from a code state $|W_{\mu}\rangle$ to a corrupted state $|W_{\mu};i\rangle$, it is essential to have the corrective jump operators in the exact form as given in the first line of Eq.~\eqref{eq:Feng_m}. On the other hand, since the no-jump evolution only accumulates an undesirable coherence between a code state $|W_{\mu}\rangle$ and a residual state $|\phi_{p}\rangle$ without direct population transfer, any preventive jump operators emptying such coherence is sufficient. Thus, we are presented with a flexibility in choosing $|\Phi_{p}\rangle$ for preventive jump operators, which may be exploited for practical experimental implementation (see, for example, the discussion below Eq.~\eqref{eq:binomial code required Hamiltonian interaction}).

In measurement-based QEC, on the other hand, the undesired coherence caused by no-jump evolution is fixed by a unitary rotation conditioned on not detecting the jump-type errors. 
The rotation angle should however be {\em fine-tuned} to precisely counter the accumulation of the
undesired coherence. 

To {make} our discussion more concrete, let us briefly review the measurement-based QEC of the binomial code outlined in \cite{BinomialCode}, and compare it with the AutoQEC scheme presented above. In measurement-based QEC of the binomial code, the error syndrome is extracted by a photon-number parity measurement \cite{Sun2014,Rosenblum2018}. If the measured parity is odd, we infer that photon loss occurred, and apply the corrective unitary operation 
\begin{align}
U_{1} &= \frac{1}{\sqrt{2}} \big{(} |0\rangle+|4\rangle \big{)}\langle 3| + |2\rangle\langle 1| + \textrm{h.c.}
\nonumber\\
&\quad + \frac{1}{2}\big{(} |0\rangle-|4\rangle \big{)} \big{(} \langle 0|-\langle 4| \big{)} , 
\end{align}
cf.~the text below Eq. (2) in \cite{BinomialCode}. On the other hand, if the measured parity is even, we infer that no loss error has happened. In this case, in order to counter the undesired accumulation of coherence between $|W_{0}\rangle = (|0\rangle+|4\rangle)/\sqrt{2}$ and $|\phi_{1}\rangle = (|0\rangle-|4\rangle)/\sqrt{2}$ (caused by the no-jump evolution), we need to apply the following corrective unitary rotation 
\begin{align}
U_{2} &= \cos(\gamma \Delta t) \big{(} |W_{0}\rangle\langle W_{0}|  + |\phi_{1}\rangle\langle\phi_{1}|  \big{)} 
\nonumber\\
&\quad + \sin(\gamma \Delta t) \big{(} |W_{0}\rangle\langle\phi_{1}| - |\phi_{1}\rangle\langle W_{0}|  \big{)} + U_{2}^{\mathrm{res}}, 
\end{align} 
where $\Delta t$ is the waiting time between the syndrome measurements and $U_{2}^{\mathrm{res}}$ is an arbitrary unitary operator on the subspace $\textrm{span}\lbrace |1\rangle,|3\rangle \rbrace$ (see Eqs. (16),(17) in \cite{BinomialCode}).

The corrective unitary operator $U_{1}$ corresponds to the corrective jump operator $F_{\mathrm{eng},1}$ in Eq.\eqref{eq:binomial code engineered dissipation}, and $U_{2}$ corresponds to the preventive jump operator $F_{\mathrm{eng},2}$. Notably, the corrective unitary operator $U_{2}$ should have the definite form above in the subspace $\textrm{span}\lbrace |0\rangle,|2\rangle,|4\rangle\rbrace$. (Especially, the rotation angle $\gamma \Delta t$ needs to be fine-tuned.) By contrast, for the preventive jump operator $F_{\mathrm{eng},2}$, we can choose any $|\Phi\rangle$ of order unity norm orthogonal to $|\phi_{1}\rangle$.


\section{Implementation of Large Engineered Dissipation} 
\subsection{Physical Requirements} 
We now discuss the practical feasibility of the $M\rightarrow \infty$ limit.
In principle, any desired Markovian engineered dissipation can be realized by modulating the Hamiltonian coupling between the system and fast decaying ancillary qubits \cite{Verstraete,Reiter2012,Zanardi2016}. 
Consider to this end the master equation
\begin{align}
\frac{d\rho_{T}(t)}{dt} = \sum_{i=1}^{L}\kappa_{i}D(|g_{i}\rangle\langle e_{i}| )\rho_{T}(t) -i[H_{\rm coupl},\rho_{T}(t)], \label{eq:master equation for engineered dissiaption}
\end{align}
with $H_{\rm coupl}=\sum_{i=1}^{L}\lambda_i  (F_{\textrm{eng},i}\otimes |e_{i}\rangle\langle g_{i}| +\textrm{h.c.}) $. Then, in the weak coupling and fast decay limit (i.e., $\lambda_i\ll \kappa_{i}$), the entire system approximately evolves as $\rho_{T}(t) = \rho(t)\otimes_{i=1}^{L}|g_{i}\rangle\langle g_{i}|$, and the system density matrix $\rho(t)$ obeys the desired Lindbladian master equation
\begin{equation}
\frac{\partial \rho}{\partial t} =  \sum_{i=1}^{L} \Big{(} \frac{ 4\lambda_i^{2}}{\kappa_{i} }\Big{)}D(F_{\mathrm{eng},i})\rho  \label{eq:master equation for the realized engineered dissiaption}
\end{equation}
(see Proposition 3 in \cite{Zanardi2016}; for the bound of deviation from the approximation, see Proposition 1 therein). The dimensionless engineered dissipation strength $M$ then scales as 
\begin{equation}
M=\mathcal{O}\Big{(} \frac{ \lambda_i^{2}}{\kappa_{i}\gamma} \Big{)}, 
\end{equation}
where $\gamma$ is the strength of intrinsic dissipation. Thus, the $M\rightarrow\infty$ limit can be achieved if $\sqrt{\kappa_{i}\gamma}\ll \lambda_i\ll \kappa_{i}$ for all $1\le i\le L$, which is feasible if $\gamma \ll \min_{1\le i\le L}\kappa_{i}$, i.e., when the ancillary qubits' decay rate is much larger than the intrinsic one of the system.

AutoQEC by engineered dissipation was experimentally realized in circuit QED systems \cite{Leghtas853,Touzard2017}. In both experiments, quantum information was 
encoded in the photon mode  of a microwave cavity, 
using the two-component cat states $|C^{\pm}_{\alpha}\rangle\propto |\alpha\rangle\pm|-\alpha\rangle$ as the logical qubit basis. Stabilization of the two-component cat code (in a high-Q cavity mode) was achieved by an engineered dissipation $D(a^{2}-\alpha^{2})$, which was implemented by coupling the high-Q cavity mode (decay rate $\gamma$) to a fast-decaying low-Q cavity mode (decay rate $\kappa$) via interaction $H_{\rm coupl} = \lambda ((a^{2}-\alpha^{2})b^{\dagger}+\textrm{h.c.})$. In the earlier experiment \cite{Leghtas853}, $\lambda \simeq$ 700\,kHz, $\kappa\simeq$ 40\,MHz, $\gamma\simeq$ 50\,kHz, and thus $M\simeq 1$ was achieved, and more recently $\lambda\simeq$ 900\,kHz, $\kappa\simeq$ 3\,MHz, $\gamma \simeq$ 10\,kHz, and thus $M\simeq 100$ was realized \cite{Touzard2017}.

We note that the realized two-component cat code is not robust against photon loss (a dominant error source in a high-Q cavity mode), since a single loss causes an irreversible logical bit-flip: $a|C^{\pm}_{\alpha}\rangle\simeq \alpha|C^{\mp}_{\alpha}\rangle$. There have been many proposals to leverage the capability of the cat code, either based on variations \cite{AutoCatTheory,Albert2018} or concatenation \cite{Cohen2017} of the two-component cat code, such that the logical information is protected against the photon loss errors. Below, we propose an alternative AutoQEC scheme, tailored to the same physical platform, based on the binomial code discussed above and briefly compare it with the cat code schemes.

\subsection{Implementation of binomial code AutoQEC} 
The engineered dissipation given in Eq.~\eqref{eq:binomial code engineered dissipation} fully protects the binomial code space against the photon loss error in the $M\rightarrow\infty$ limit (see Fig.~\ref{fig:Preventive needed}). In principle, such an engineered dissipation can be realized by coupling a high-Q cavity mode to two fast-decaying transmon qubits \cite{Koch2007,Schreier2008} or low-Q cavity modes via the interaction
\begin{align}
H_{1} &= \lambda_{1} \Big{(} \frac{1}{\sqrt{2}}(|0\rangle+|4\rangle)\langle 3| +|2\rangle\langle 1| \Big{)} \otimes |e_{1}\rangle\langle g_{1}| + \textrm{h.c.} ,
\nonumber\\
H_{2} &= \lambda_{2} ( |\Phi\rangle\langle 0|-|\Phi\rangle\langle 4| )\otimes |e_{2}\rangle\langle g_{2}| +\textrm{h.c.} , \label{eq:binomial code required Hamiltonian interaction}
\end{align}  
where $|g_{i}\rangle$ and $|e_{i}\rangle$ are the ground and excited states of the $i^{\textrm{th}}$ qubit (or low-Q cavity mode) for $i=1,2$ (cf.~Eqs. \eqref{eq:master equation for engineered dissiaption},\eqref{eq:master equation for the realized engineered dissiaption}), respectively. 

The ineraction $H_{1}$ requires selective two-quanta exchanges, $|4\rangle\langle 3|\otimes |e_{1}\rangle\langle g_{1}|$ and $|2\rangle\langle 1|\otimes |e_{1}\rangle\langle g_{1}|$, and a four-quanta exchange $|0\rangle\langle 3|\otimes |e_{1}\rangle\langle g_{1}|$. Also, $H_{2}$ requires selective three-quanta exchanges $|2\rangle\langle 0| \otimes |e_{2}\rangle\langle g_{2}|$ and $|2\rangle\langle 4| \otimes |e_{2}\rangle\langle g_{2}|$ if we choose $|\Phi\rangle = |2\rangle$. (Note that $|\Phi\rangle = |2\rangle$ allows the lowest order interaction, and other choices for $|\Phi\rangle$ lead us to at least four-quanta exchanges.) In comparison, variations of the cat code require an engineered four-photon dissipation $D(a^{4}-\alpha^{4})$ (four-component cat code \cite{AutoCatTheory}) or $D(a_{1}^{2}a_{2}^{2}-\gamma^{4})$ (two-mode cat code \cite{Albert2018}), hence at least a five-quanta exchange with an ancillary mode $b$, i.e., $a^{4}b^{\dagger}$ or $a_{1}^{2}a_{2}^{2}b^{\dagger}$. However, these interactions are not selective with regard to the photon number present, whereas the binomial code requires exactly these photon-number-selective quanta exchanges.

Generating higher (than second) order interaction between a cavity mode and a qubit (or another cavity mode) is a challenging task (especially at strong coupling), 
both in the photon number selective and non-selective cases. We however remark that photon-number {\em non-selective} three quanta exchange $a^{2}b^{\dagger}$ was realized in two previous experiments \cite{Leghtas853,Touzard2017}, and there is a concrete theoretical proposal for the photon number non-selective six-quanta exchange $a^{4}b^{\dagger 2}$ \cite{Mundhada2017}.

\section{Conclusion} 
In summary, 
we have provided a sufficient condition of the Knill-Laflamme type, 
under which the code space manifold can be autonomously protected against intrinsic errors in the limit of large engineered dissipation.  
We have constructed explicit engineered jump operators to achieve 
AutoQEC, and derived the temporal error bounds up to which the information initially 
stored in the code space can be preserved. We have briefly compared AutoQEC with measurement-based QEC, and clarified the difference between them as regards the protection against no-jump evolution. Based on our general theory
of AutoQEC, we have proposed an autonomous implementation of the binomial code, which may allow for significantly longer logical qubit lifetimes than those of the most stable uncorrected physical qubits.


\acknowledgments  
We gratefully acknowledge Julien M. E. Fra\"{\i}sse  for a critical reading and the corresponding comments.  
KN thanks  Liang Jiang, Steven M. Girvin, Michel Devoret, Steven T. Flammia, Victor V. Albert, Shruti Puri, Shantanu Mundhada, Connor Hann and Salvatore Elder for fruitful discussions.
The research of JML and URF 
was supported by the  National Research Foundation of Korea (NRF),
, Grant No. 2017R1A2A2A05001422. KN acknowledges support through the 
 Korea Foundation for Advanced Studies (KFAS). 

\appendix 




\section{Decomposition of dissipative evolution and noiseless subsystem symmetry}  \label{AppA}
Let the code space $\mc{C}=\textrm{span}\lbrace |W_{\mu}\rangle\rbrace_{\mu=0,\cdots,d-1}$ satisfy the Knill-Laflamme condition 
Eq.~\eqref{eq:k-l}. 


In the intrinsic dissipation operator, we identify two separate parts by using the projection operator 
$\Pi_{\mc{C}}$: 
\begin{equation}
F_k = F_k \Pi_{\mc{C}} + F_k(1-\Pi_{\mc{C}}) = f_k + \tilde{f}_{k}. \label{eq:Ff}
\end{equation}
The part $f_k =F_k \Pi_{\mc{C}}$ is the component of $F_k$ that is applied on the code subspace. 
We first show that, if the remaining term $\tilde{f}_{k}$ is absent, all $\mc{S}_{\mu}$ are collecting subspaces: For each $\mu\in\lbrace 0,\cdots,d-1 \rbrace$, an arbitrary initial state in $\mc{S}_{\mu} \subseteq \mathcal{H}$ never leaves $\mc{S}_{\mu}$ during the time evolution. This holds 
because for each Lindblad operator $\hat{F} \in \{ f_k, F_{\textrm{eng} ,i} \} \coloneqq \mathcal{F}$, the following general 
conditions hold (see Lemmas 9 and 11 of \cite{Baumgartner2008}) 
\begin{gather}
    \hat{F} \Pi_{\mc{S}_{\mu}} = \Pi_{\mc{S}_{\mu}} \hat{F} \Pi_{\mc{S}_{\mu}} , \label{eq:F-cond a}\\
    \Pi_{\mc{S}_{\mu}} \Big{(}  iH - 
    \sum_{\hat{F} \in \mathcal{F}} \hat{F}^\dagger \hat{F} \Big{)}
    (1-\Pi_{\mc{S}_{\mu}}) = 0 , \label{eq:F-cond b}
\end{gather}
where $\Pi_{\mc{S}_{\mu}}$ is the projection operator onto $\mc{S}_{\mu}$ and $H=0$ in our case since we set $H_{\textrm{eng}} = -H_{\textrm{sys}}$.  
It then follows that each $\mathcal{S}_{\mu}$ supports a unique steady state, $\rho_{\mu,\textrm{st}}$ for each $\mu\in \lbrace 0,\cdots,d-1\rbrace$.

In addition to these $d$ steady states, there exist $d^{2}-d$ stationary phase relations between all $\mathcal{S}_{\mu}$ and $\mathcal{S}_{\nu}$ with $\mu\neq \nu$: Note that the following unitary \textit{intertwiner} 
\begin{equation} \label{eq:unitary}
	U_{\mu}= |W_{\mu}\rangle \langle W_{0}| +\sum_{i=1}^{m-1} |W_{\mu};i\rangle \langle W_{0};i| + \textrm{h.c.} 
	+ \sum_{p}|\phi_p \rangle \langle \phi_p |, 
\end{equation} 
satisfies $U_{\mu}\Pi_{\mc{S}_{0}}=\Pi_{\mc{S}_{\mu}}U_{\mu}$, $U_{\mu}^{2}=\mathds{1}$. It commutes with all $\hat{F} \in \mathcal{F}$, due to $\mu$-independence of $\langle W_{\mu};i|F_{k}|W_{\mu}\rangle$.  Then, Proposition 16 of \cite{Baumgartner2008} implies that all steady states $\rho_{\mu,\textrm{st}}$ are unitarily connected by $U_{\mu}$, i.e., $\rho_{\mu,\textrm{st}} = U_{\mu}\rho_{0,\textrm{st}}U_{\mu}^{\dagger}$, and furthermore $U_{\mu}\rho_{0,\textrm{st}}U_{\nu}^{\dagger}$ establishes a stationary phase coherence between $\mathcal{S}_{\mu}$ and $\mathcal{S}_{\nu}$ for $\mu\neq\nu$ -- 
cf.~Fig.~\ref{fig:schematic}. Therefore, for any $d\times d$ qudit density matrix $\omega$, 
\begin{equation}
\Omega_{\textrm{st}}=\sum_{\mu,\nu=0}^{d-1}\omega_{\mu\nu}U_{\mu}\rho_{0,\textrm{st}}U^{\dagger}_{\nu} \leftrightarrow \omega\otimes \rho_{0,\mathrm{st}} \label{eq:steady states when the error is limited to the code subspace}
\end{equation}
are left invariant under the Lindbladian dynamics generated by the jumps in $\mathcal{F}$. Thus $\mathcal{S}_{\textrm{ccs}}=\cup_{\mu=0}^{d-1}\mathcal{S}_{\mu}$ forms a noiseless subsystem \cite{Zanardi1997,Lidar1998,Knill2000,Kempe2001,Lidar2012} if the intrinsic dissipation operator is limited to act on the code subspace. Note that states outside the {corrupted code space} are omitted in Fig.~\ref{fig:schematic} and Eq.~\eqref{eq:steady states when the error is limited to the code subspace}, since they are not occupied in the steady states.

\section{Long-time limit for large engineered dissipation} \label{AppB}
Given the noiseless subsystem symmetry derived in Appendix \ref{AppA}, 
it is natural to divide the Lindbladian $\mathcal{L}$ in (\ref{eq:master}) into three parts as follows 
\begin{eqnarray}
M\mc{L}_0[\rho] &=& M\sum_{i=1}^{L} D(F_{\textrm{eng},i})\rho ,\quad \mc{L}_1 [\rho] =\sum_{k=1}^{N} D(f_k)\rho, 
\nonumber\\
\mc{L}_2 [\rho] &=& \sum_{k=1}^{N}\Big{(} D(F_k)-D(f_k) \Big{)}. \label{decomp}
\end{eqnarray}
In the above, $\mc{L}_0$ is the engineered part while $\mc{L}_1$ is the intrinsic dissipation that follows the noiseless subsystem symmetry and leaves $\omega \otimes \rho_{\mathrm{st}}$ invariant: $\left(M\mc{L}_0+\mc{L}_1\right)[\omega \otimes \rho_{\mathrm{st}}]= 0$. Finally, $\mc{L}_2$ breaks this noiseless subsystem symmetry and brings about evolution for the otherwise steady state, $\mc{L}_2[\omega \otimes \rho_{\mathrm{st}}] \neq 0$. 
Following the decomposition in \eqref{decomp}, we define new Lindbladians, with properly chosen scaling factors 
\begin{equation}    \label{eq:lindbladian}
    \frac{1}{M}\mc{L} = \mc{L}_0 + \frac{1}{M}\mc{L}_1 + \frac{1}{M}\mc{L}_2 \coloneqq\mc{L}_{\textrm{e}} + \frac{1}{M}\mc{L}_2, 
\end{equation}
and derive the main result Eq.~\eqref{eq:thm} by  
proving an equivalent statement that a time evolution generated by the rescaled Lindbladian 
$\mc{L}/M$ leaves states in $\mc{C}$ approximately invariant in the same sense up to time $MT$.
We use the perturbation theory of Lindbladian superoperators \cite{kato1995perturb}, which has found applications so far  when the magnitude of a weakly unitary or dissipative operator is the small parameter of the perturbation series \cite{Zanardi2014, Zanardi2015, Zanardi2016, Macieszczak2016}. In our approach, on the other  hand, the inverse of the magnitude of strong engineered dissipation, $1/M$, is the small parameter.
Formally, we establish the following bound of the error $\epsilon(t)$ in the interval $0 \leq t \leq MT$, for states projected to the code subspace $\mc{C}$: 
\begin{eqnarray} \label{eq:bound_err}
    \epsilon(t) &\leq & \opnorm{e^{t \mc{L} /M}\mc{P}_{\mathcal{C}} - \mc{P}_{\mathcal{C}}} \\
    &\leq& \frac{T}{M} \left[ \mc{O}(\opnorm{\mc{L}_2}) + \mc{O}(\tau\opnorm{\mc{L}_2}^2) \right] + \frac{1}{M}\mc{O}(\tau \opnorm{\mc{L}_2}) \nonumber, 
\end{eqnarray}
where $\mathcal{P}_{\mathcal{C}}\coloneqq\lim_{t\rightarrow \infty} e^{t\mathcal{L}_{0}}$ is the superoperator projection to the code space $\mathcal{C}$ and $\tau=\mc{O}(M^0)$ is of order the inverse of the smallest real part of nonzero eigenvalue of ${\mc{L}_{\textrm{e}}}/M$, as discussed below. 
The above bound  \eqref{eq:bound_err}  
leads to our AutoQEC Theorem as represented by 
Eq.~\eqref{eq:thm}. 

We now proceed to prove the bound in the second line of Eq.~\eqref{eq:bound_err}. 
First, we denote the rescaled Lindbladian superoperator as $\overline{\mc{L}} = \mc{L}/M$
and, similarly, $\overline{\mc{L}_2}$ and $\overline{\mc{L}_\mathrm{e}}$. 
The error $\epsilon (t)$ as defined in Eq.~\eqref{eq:bound_err} satisfies the following inequalities: 
\begin{widetext}
\begin{eqnarray}    \label{eq:error}
    \epsilon(t) &=&\parallel e^{t\overline{\mc{L}}}\mc{P}_{\textrm{C}} - \mc{P}_{\textrm{C}} \parallel \nonumber \\
    &\leq& \parallel e^{t\mc{P}\overline{\mc{L}}\mc{P}}\mc{P}_{\textrm{e}} - \mc{P}_{\textrm{e}} \parallel + \parallel ( e^{t\overline{\mc{L}}} - e^{t\mc{P}\overline{\mc{L}}\mc{P}} ) \mc{P}_{\textrm{e}} \parallel 
+ \left( \opnorm{ e^{t\overline{\mc{L}}} } + 1 \right) \opnorm{\mc{P}_{\textrm{C}} - \mc{P}_e} \nonumber   \\
    &\leq& \parallel e^{t\mc{P}\overline{\mc{L}}\mc{P}}\mc{P}_{\textrm{e}} - \mc{P}_{\textrm{e}} \parallel 
    + (\opnorm{e^{t\overline{\mc{L}}}} + \opnorm{e^{t\mc{P}\overline{\mc{L}}\mc{P}}} ) \parallel \mc{P} - \mc{P}_{\textrm{e}}  \parallel 
    + \left( \opnorm{ e^{t\overline{\mc{L}}} } + 1 \right) \opnorm{\mc{P}_{\textrm{C}} - \mc{P}_e} .
\end{eqnarray}
\end{widetext}
The triangle inequalities $\opnorm{A+B} \leq \opnorm{A} + \opnorm{B}$ and $ ( e^{t\overline{\mc{L}}} - e^{t\mc{P}\overline{\mc{L}}\mc{P}} )\mc{P} = 0$ were used.

The kernel of $\mc{L}_{\textrm{e}}$ is $d^{2}$-dimensional because we have a $d$-dimensional noiseless subsystem. Let $\mc{P}$ be the direct sum of the projection onto eigenspaces of $\mc{L}$, which perturbatively originate from the kernel of the unperturbed superoperator $\mc{L}_{\textrm{e}}$. Then, from the perturbation theory of linear operators \cite{kato1995perturb}, 
\begin{widetext}
\begin{eqnarray}\label{eq:perturb_P}
    &\mc{P} - \mc{P}_{\textrm{e}} = -\left( \mc{P}_{\textrm{e}} \overline{\mc{L}_2} \mc{S} - \mc{S} \overline{\mc{L}_2} \mc{P}_{\textrm{e}} \right) + \mc{O}(\parallel \overline{\mc{L}_2} \parallel ^2), \\ \label{eq:perturb_L}
    & \mc{P}\overline{\mc{L}}\mc{P} = \left( \mc{P}_{\textrm{e}}\overline{\mc{L}_2}\mc{P}_{\textrm{e}} \right) 
    -  \left( \mc{P}_{\textrm{e}}\overline{\mc{L}_2}\mc{S}\overline{\mc{L}_2}\mc{P}_{\textrm{e}} 
    - \mc{P}_{\textrm{e}}\overline{\mc{L}_2}\mc{P}_{\textrm{e}}\overline{\mc{L}_2}\mc{S}
    - \mc{S}\overline{\mc{L}_2}\mc{P}_{\textrm{e}}\overline{\mc{L}_2}\mc{P}_{\textrm{e}} \right) 
    + \mc{O}(\tau^2 \opnorm{\overline{\mc{L}_2}}^3)
\end{eqnarray}     
\end{widetext} 
holds. Here, $\mc{S}$ is the pseudo-inverse of $\overline{\mc{L}_{\textrm{e}}}$, which satisfies $\mc{S}\overline{\mc{L}_{\textrm{e}}} = \overline{\mc{L}_{\textrm{e}}}\mc{S} = 1-\mc{P}_{\textrm{e}}$. 
Then, $\tau = \opnorm{\mc{S}}$ is of the order of the inverse of the smallest nonzero eigenvalue of $\overline{\mc{L}_{\textrm{e}}}$ and 
is of order $\mc{O}(M^0)$. This is proven as follows: 
$\mc{P}_{\textrm{C}},~\mc{P}_{\textrm{e}}$ denote projection superoperators that project the linear operators of the Hilbert space $\mc{H}$ onto the kernel of the superoperators $\mc{L}_0,~\mc{L}_{\textrm{e}}$, respectively. For example, $\mc{P}_C$ projects operators onto the kernel $\textrm{ker}\, \mc{L}_0 = \{ \rho : \mc{L}_0[\rho] = 0 \}$, which is indeed the code subspace. Since $\cup_{\mu=0}^{d-1}\mathcal{S}_{\mu}$ forms a noiseless subsystem of $\mc{L}_{\textrm{e}}$, Appendix \ref{AppA}, $\textrm{ker}\, \mc{L}_{\textrm{e}}$ is a projection onto the subspace composed by metastable states, $\omega \otimes \rho_{\textrm{st}}$.
Using the perturbation theory of linear operators viewing $\mc{L}_1/M$ as a perturbation to the  Lindbladian $\mc{L}_{\textrm{e}} = \mc{L}_0 + \mc{L}_1/M$, one finds \cite{kato1995perturb}
\begin{equation} \label{eq:sup1:proj}
\mc{P}_{\textrm{e}} - \mc{P}_C = -\frac{1}{M} \left( \mc{P}_C \mc{L}_1 \mc{S}_0 - \mc{S}_0 \mc{L}_1 \mc{P}_C \right) + \frac{1}{M^2} \mc{O}(\parallel \mc{L}_1 \parallel ^2),
\end{equation}
where $\mc{S}_0$ denotes the pseudo-inverse of $\mc{L}_0$, which satisfies $\mc{S}_0\mc{L}_0 = \mc{L}_0\mc{S}_0 = 1-\mc{P}_C$.
This implies $\opnorm{ \mc{P}_{C} - \mc{P}_{\textrm{e}} } = \mc{O}(1/M)$.

Now, by the definition of \unexpanded{$\widetilde{f_k}$: $\widetilde{f_k}\Pi_\mc{C} = 0$}, it follows that $\mc{P}_{C}\mc{L}_2 \mc{P}_{C}[\rho] = \Pi_\mc{C}\mc{L}_2 [\Pi_\mc{C} \rho \Pi_\mc{C}]\Pi_\mc{C} = 0$. This entails $\opnorm{\mc{P}_{C}\mc{L}_2 \mc{P}_{C}} = 0$.
Putting the results together using the triangle inequality $\opnorm{A+B} \leq \opnorm{A} + \opnorm{B}$, we get
\begin{eqnarray} 
\opnorm{\mc{P}_{\textrm{e}}\mc{L}_2 \mc{P}_{\textrm{e}}}
&\leq& \opnorm{ \left( \mc{P}_{\textrm{e}} - \mc{P}_C \right) \mc{L}_2 \mc{P}_{\textrm{e}}} \nn 
& & \hspace*{-6em} + \opnorm{ \mc{P}_{\textrm{C}} \mc{L}_2 \left( \mc{P}_{\textrm{e}} - \mc{P}_{\textrm{C}} \right) }+ \opnorm{\mc{P}_{C}\mc{L}_2 \mc{P}_{C}} ~=~\mc{O}(\frac{1}{M}). \label{eq:sup1:PeL2Pe}
\end{eqnarray}

From Eqs.~(\ref{eq:sup1:PeL2Pe}),(\ref{eq:perturb_P}),(\ref{eq:perturb_L}), we obtain two scaling equations for the norm of the operators therein,  
\begin{eqnarray}    \label{eq:bound_proj}
    \parallel \mc{P} - \mc{P}_{\textrm{e}} \parallel  &=&\mc{O}(\tau \opnorm{\overline{\mc{L}_2}}  ), \\
  \label{eq:bound_op}
	\opnorm{\mc{P}\overline{\mc{L}}\mc{P}} &=& \frac{1}{M}\mc{O}(\opnorm{\overline{\mc{L}_2}}) + \mc{O}(\tau\opnorm{\overline{\mc{L}_2}}^2).
\end{eqnarray}
Since $\mc{L}$ by definition generates a completely positive trace preserving evolution, $\parallel \!\!e^{t\overline{\mc{L}}} \!\! \parallel \leq 1$. Furthermore, because time is bounded by $t \leq MT$, we have $t \opnorm{\mc{P}\overline{\mc{L}}\mc{P}} \leq \frac{T}{M} \left( \mc{O}(\opnorm{\mc{L}_2}) + \mc{O}(\tau\opnorm{\mc{L}_2}^2) \right)$. Note that here the overline has been removed. Then, $t\opnorm{\mc{P}\overline{\mc{L}}\mc{P}}$ vanishes for sufficiently large $M$, hence $\opnorm{e^{t\mc{P}\overline{\mc{L}}\mc{P}}} = \mc{O}(1)$.
Using this bound and $\opnorm{e^X - 1} \leq \opnorm{X} \opnorm{e^X}$ with $X = t\mc{P}\overline{\mc{L}}\mc{P}$, it follows that 
\begin{equation}    \label{eq:bound_expop}
    \opnorm{e^{t\mc{P}\overline{\mc{L}}\mc{P}}\mc{P}_e - \mc{P}_e} \leq \frac{T}{M} \left( \mc{O}(\opnorm{\mc{L}_2}) + \mc{O}(\tau\opnorm{\mc{L}_2}^2) \right).
\end{equation}
Substituting Eqs.~\eqref{eq:bound_proj}, \eqref{eq:bound_op}, \eqref{eq:bound_expop} into 
Eq.~(\ref{eq:error}), the bound on the error reads as follows:
\begin{equation}    \label{eq:supp_bound_err}
    \epsilon(t) \leq \frac{T}{M} \left[ \mc{O}(\opnorm{\mc{L}_2}) + \mc{O}(\tau\opnorm{\mc{L}_2}^2) \right] + \frac{1}{M}\mc{O}(\tau \opnorm{\mc{L}_2}).
\end{equation}
This completes our proof of Eq.~\eqref{eq:bound_err}.

\bibliography{jaemo36}

\end{document}